\documentclass{ifacconf}
\usepackage{graphicx}      
\usepackage{natbib}        
\usepackage{amsmath,amssymb,amsfonts}
\usepackage{mathtools}
\usepackage{mathrsfs}
\usepackage{bm}
\usepackage{algpseudocode} 
\usepackage[algo2e,vlined,ruled]{algorithm2e}

\usepackage{tikz-cd}
\usepackage{tikz}
\usetikzlibrary{arrows.meta, positioning, calc}
\usetikzlibrary{decorations.pathreplacing}
\usepackage{etoolbox}
\usepackage{arydshln}
\AtBeginDocument{

 \newtheorem{newassum}{Assumption}
 \renewenvironment{assum}{\begin{newassum}}{\end{newassum}}
  \newtheorem{newthm}{Theorem}
 \renewenvironment{thm}{\begin{newthm}}{\end{newthm}}
 \newtheorem{newlem}{Lemma}
 \renewenvironment{lem}{\begin{newlem}}{\end{newlem}}
 \newtheorem{newdefn}{Definition}
 \renewenvironment{defn}{\begin{newdefn}}{\end{newdefn}}
  \newtheorem{newrem}{Remark}
 \renewenvironment{rem}{\begin{newrem}}{\end{newrem}}
}

\allowdisplaybreaks
\newcommand{\R}{\mathbb{R}}

\newcommand{\rank}{\mathrm{rank}\, }
\newcommand{\kernel}{\mathrm{Ker}\, }
\newcommand{\image}{\mathrm{Im}\, }

\newcommand{\X}{\mathscr{X}}

\newcommand{\W}{\mathscr{W}}

\newcommand{\col}{\mathrm{col}}
\newcommand{\matrices}[1]{\begin{bmatrix} #1\end{bmatrix}}

\begin{document}
\begin{frontmatter}

\title{
Distributed State Estimation for Discrete-Time Systems With Unknown Inputs: An Optimization Approach
} 


\author[First]{Ruixuan Zhao} 
\author[Second]{Guitao Yang} 
\author[Third]{Nicola Bastianello}
\author[First]{Boli Chen}

\address[First]{Department of Electronic and Electrical Engineering, and Dynamic Systems Lab, University College London, London, WC1E 7JE, UK (e-mail: \{ruixuan.zhao.22,boli.chen\}@ucl.ac.uk).}
\address[Second]{Wolfson School of Mechanical, Electrical and Manufacturing Engineering,
    Loughborough University, 
    Loughborough, LE11 3TU, UK 
    (e-mail: g.yang@lboro.ac.uk).}
\address[Third]{School of Electrical Engineering and Computer Science, and Digital Futures, KTH Royal Institute of Technology, Stockholm, Sweden, (e-mail: nicolba@kth.se).}

\begin{abstract}
This paper proposes a novel Distributed Unknown Input Observer (DUIO) framework for state estimation in large-scale systems subject to local unknown inputs. We consider systems where outputs are measured by a network of spatially distributed sensors and inputs are introduced through multiple dispersed channels. In this framework, each local node utilizes only its local input and output measurements to estimate the maximal locally reconstructible state. Subsequently, nodes collaboratively reconstruct the whole system state via a distributed optimization algorithm that fuses these partial estimates. We provide a rigorous analysis showing that the estimation error is bounded, with the error bound explicitly dependent on the number of communication iterations per time step and strongly convexity constant determined by the system parameters. Furthermore, to counteract curvature anisotropy induced by poor conditioned system geometry, we embed a normalization step into the distributed optimization procedure. Simulation results demonstrate the effectiveness of the proposed framework and the performance improvements yielded by the normalization procedure.
\end{abstract}


\end{frontmatter}

\section{Introduction}
With the rapid expansion of cyber-physical systems (CPS) and the internet of things (IoT), the scale of modern engineering infrastructures, such as power grids, water distribution networks, and intelligent transportation systems, has grown significantly. In such interconnected infrastructures, real-time state estimation is critical for monitoring system status, detecting faults or cyber-attacks, and facilitating decentralized control decisions  \citep{liu2023distributed, kim2023decentralized, fioravanti2024distributed}. While centralized estimation schemes have traditionally served as the standard approach, they are increasingly ill-suited for large-scale networks due to their vulnerability to single-point failures and communication bottlenecks. Consequently, distributed state estimation (DSE) has emerged as a more reliable and scalable paradigm, enabling a network of sensors to collaboratively reconstruct the global system state using only local measurements and neighbor-to-neighbor communication \citep{rego2019distributed}.

A fundamental challenge in deploying DSE in realistic scenarios is the presence of unknown inputs, such as unmodeled dynamics, external disturbances, or malicious cyber-attacks. In a distributed setting, this problem is exacerbated by the limited observability of local nodes. A single node typically lacks access to all system inputs and may not have sufficient information to distinguish between nominal state evolution and external disturbances. Existing solutions (e.g., \cite{yang2022state}) often circumvent this difficulty by imposing stringent local rank conditions, requiring that unknown inputs be locally decoupled or reconstructed at each node. This assumption is restrictive and often impractical for sparse large-scale networks, as it essentially requires every node can individually address the unknown input, thereby undermining the cooperative nature of distributed estimation.

The challenge is further compounded in the discrete-time domain. Unlike continuous-time observers that can exploit high-gain mechanisms to achieve asymptotic convergence or sliding modes to fully reject disturbances, discrete-time designs face inherent limitations in consensus dynamics and disturbance decoupling. While geometric approaches have been proposed to address state estimation with unknown inputs, they have traditionally focused on centralized designs \citep{bhattacharyya2003observer}. Optimization-based approaches, such as distributed moving horizon estimation (MHE) (e.g., \cite{farina2010distributed}), offer a promising alternative for discrete-time systems. They naturally handle constraints and offer a flexible structure that is easier to deploy and extend to complex scenarios, including nonlinear dynamics. However, existing distributed optimization strategies typically assume the complete absence of unknown inputs or rely on strict local observability conditions. There remains a significant gap in the literature for a distributed, optimization-based framework that can robustly estimate the state of a discrete-time system in the presence of unknown inputs without relying on strong local observability.




More recently, distributed unknown input observer (DUIO) designs with asymptotic convergence guarantees have attracted attention, including \cite{yang2022state, cao2023distributed, zhang2025distributed, wu2025almost}. However, these designs focus on continuous-time systems and adopt high-gain injection to enforce consensus, which may compromise robustness against measurement or communication noise. Furthermore, to locally reject the influence of unknown inputs, these designs require a specific rank condition at each node, which is restrictive for real-world implementation. 
The framework presented by \cite{disaro2025distributed} offers an efficient single timescale architecture for discrete-time systems but is still limited by the requirement of a local rank condition for disturbance decoupling, restricting its use in networks with limited local observability. Although \cite{zhao2025bridging, zhao2025DUIO} address this limitation by relaxing the rank condition, their solutions are tailored to continuous-time dynamics and cannot be directly extended to the discrete-time domain.

Apart from the development that involves DUIO designs, distributed optimization-based estimation schemes have drawn increasing attention (e.g., \cite{wang2024distributed, yang2025state}). These designs are prone to address constraints of real IoT systems such energy consumption and limited bandwidth, and extend to nonlinear systems \cite{wang2024distributed}. However, current optimization-based designs generally assume autonomous systems or scenarios where all inputs are known to all nodes, thereby failing to account for unknown input channels.


Despite all this progress, there remains a gap for a generalizable discrete-time DSE framework that handles unknown inputs without requiring restrictive local rank conditions. This paper addresses this gap via a distributed optimization approach. The main contributions of this paper are summarized as follows: 1) we introduce a novel distributed optimization-based state estimation scheme for discrete-time LTI systems that explicitly accounts for unknown inputs and eliminates the restrictive rank condition; 2) Compared with existing optimization-based DSE frameworks \cite{wang2024distributed}, the proposed method addresses unknown inputs and achieves delay-free estimates. In contrast to \cite{yang2025state}, our proposed method requires sharing only local estimated states, thereby eliminating the need to transmit local gradient information; 3) we propose a normalization step to circumvent the curvature anisotropy induced by poor conditioned system geometry.

    


The remainder of this paper is organized as follows. In Section~\ref{sec:problem statement}, we formulate the problem of distributed state estimation with local unknown inputs and define the discrete-time DUIO. Section~\ref{sec:DUIO} presents our novel DUIO design based on a distributed optimization method and establishes its convergence properties. Finally, we validate the effectiveness of the proposed design via a numerical example in Section~\ref{sec:simulate} and provide concluding remarks in Section~\ref{sec:conclude}.

\paragraph*{Notation}
The sets of real numbers are denoted by $\mathbb{R}$. The identity matrix of size $n$ is denoted by $I_{n}$. The all-zeros vector of size $n$ is denoted by $\mathbf{0}_{n}$. $\mathbf{1}_n$ is defined as a $n$-dimensional vector with all $1$ elements. The symbol $\|\cdot\|$ denotes the standard Euclidean norm, and $\otimes$ denotes the Kronecker product. $\kappa(M)$ denotes the spectrum of $M$ and $\lambda_{\mathrm{min}}(M)$ represents the minimal eigenvalue of $M$.\\
The communication topology of a network is represented by an undirected graph denoted by $\mathcal{G} = (\mathbf{N},\mathcal{E},\mathcal{A})$, where $\mathbf{N}= \{1,2,\dots,N\}$ is a finite nonempty set of nodes of the graph (describing the networked observer containing $N$ local sensors), $\mathcal{E}\subseteq \mathbf{N} \times \mathbf{N}$ represents the edges of the graph (describing communication among the nodes), and $\mathcal{A}=[a_{ij}] \in \mathbb{R}^{N\times N}$ is the adjacency matrix, where $a_{ij}=a_{ji}=1$ if there exists an edge between node $i$ and node $j$, and $a_{ij}=a_{ji}=0$ otherwise. The set $\mathcal{N}_i\coloneq\{j|(j,i)\in\mathcal{E}\}$ denotes the neighboring nodes of node $i$. $d_i=\sum_{i=1}^N a_{ij}$ is the degree of node $i$. $\mathcal{D}=\mathrm{diag}_{i\in\mathbf{N}}(d_i)$ denotes the degree matrix of graph $\mathcal{G}$. $\mathcal{L}=\mathcal{D}-\mathcal{A}$ is defined as the Laplacian matrix associated with graph $\mathcal{G}$. For a connected graph, $\lambda_2(\mathcal{L})$ is the smallest positive eigenvalue of $\mathcal{L}$.

\section{Problem Formulation}
\label{sec:problem statement}
Consider the following discrete-time linear time-invariant (LTI) system
\begin{subequations}\label{eq:sys}
    \begin{align}
        &x(t+1) = Ax(t) + Bu(t),\\
        &y(t) = C x(t),
    \end{align}
\end{subequations}
where ${A\in \mathbb{R}^{n\times n}}$, ${B\in \mathbb{R}^{n\times m}}$ and $C\in\mathbb{R}^{p\times n}$ are system matrices, ${x \in \mathbb{R}^n}$ is the state vector, ${u\in \mathbb{R}^m}$ is the control input, and $y\in\mathbb{R}^p$ is the output measurement.\\
We consider the case of a large-scale system, where the output is measured by a group of sensors:
    \begin{align}
        y_i(t)=C_i x(t),\ i\in\mathbf{N}
    \end{align}
where $y_i\in\mathbb{R}^{p_i}$, $\sum_{i=1}^N p_i=p$ with $C=\mathrm{col} (C_1,C_2,\cdots,C_N)$, and $y=\mathrm{col}(y_1,y_2,\cdots,y_N)$.\\
At each node $i$, the input term of the system can be partitioned by
    \begin{equation}\label{eq:inputdecomp}
    Bu(t) = B_i u_i(t) +\bar{B}_i \bar{u}_i(t),
    \end{equation}
with $B_i \in \mathbb{R}^{n\times l_i}$, $\bar{B}_i \in \mathbb{R}^{n\times (m - l_i)}$, where $m-l_i \leq p_i \leq n$. $u_i \in \mathbb{R}^{l_i}$ and $\bar u_i \in \mathbb{R}^{m - l_i}$ are the known and unknown input signals for node $i$, respectively. 
The unknown input represents all exogenous and unmeasured signals acting on the plant, including disturbances, model discrepancies, and any locally injected inputs that are not available to the observer.  
\begin{assum}\label{asm:connected}
    The undirected graph $\mathcal{G} = (\mathbf{N},\mathcal{E},\mathcal{A})$ that describes the communication connection among the local DUIOs is connected.    
\end{assum}
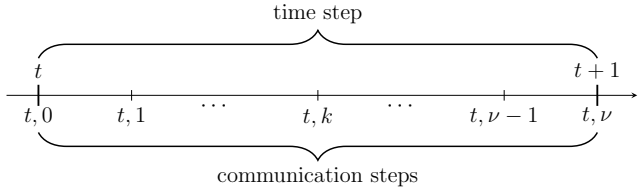
\begin{figure}[htp] 
\centering
\scalebox{0.7}{
    \begin{tikzpicture}[
    >=stealth,              
    x=1.5cm,            
    font=\large,
    tick style/.style={thin} 
]

\draw[->] (-0.4, 0) -- (7.5, 0);

\pgfmathsetmacro{\xT}{0}
\pgfmathsetmacro{\xTplusOne}{7}

\draw[line width=1pt] (\xT, 0.2) -- (\xT, -0.2);
\node[above] at (\xT, 0.2) {$t$};
\draw[line width=1pt] (\xTplusOne, 0.2) -- (\xTplusOne, -0.2);
\node[above] at (\xTplusOne, 0.2) {$t+1$};

\pgfmathsetmacro{\stepsize}{(\xTplusOne - \xT) / 6} 

\foreach \i in {0, 1, 3}
{
    \pgfmathsetmacro{\xval}{\xT + \i * \stepsize}
    \draw[tick style] (\xval, 0.1) -- (\xval, -0.1); 
    \ifnum\i=0
        \node[below] at (\xval, -0.1) {$t, 0$};
    \fi
    \ifnum\i=1
        \node[below] at (\xval, -0.1) {$t, 1$};
    \fi
    \ifnum\i=3
        \node[below] at (\xval, -0.1) {$t, k$};
    \fi
}

\node at (\xT + 1 * \stepsize + 0.9 * \stepsize, -0.2) {$\cdots$};

\node at (\xT + 3 * \stepsize + 0.9 * \stepsize, -0.2) {$\cdots$};

\pgfmathsetmacro{\xnuMinusOne}{\xT + 5 * \stepsize}
\pgfmathsetmacro{\xnu}{\xTplusOne}

\draw[tick style] (\xnu-7/6, 0.1) -- (\xnu-7/6, -0.1);
\node[below] at (\xnu-7/6, -0.1) {$t, \nu-1$};

\node[below] at (\xnu, -0.1) {$t, \nu$};

\draw [
    decorate,
    decoration={
        brace,
        amplitude=15pt,      
        mirror,             
        raise=20pt          
    },
    line width=0.8pt
] (\xT, 0) -- (\xTplusOne, 0)
node [midway, above=35pt] {\text{time step}}; 

\draw [
    decorate,
    decoration={
        brace,
        amplitude=15pt,
        raise=20pt          
    },
    line width=0.8pt
] (\xT, 0) -- (\xTplusOne, 0)
node [midway, below=35pt] {\text{communication steps}}; 

\end{tikzpicture}}\\[-1.5ex]
    \caption{Time step of system dynamics and communication steps of the algorithm. $\nu$ is the total communication steps between $t$ and $t+1$.}
    \label{fig:time_step}
\end{figure}
The objective is to design a DUIO $\{\mathcal{O}_i\}_{i\in\mathbf{N}}$ to jointly reconstruct the state vector $x$ at each node, while each node $i\in\mathbf{N}$, has only access to its local known control input $u_i$, local measurement $y_i$ and the shared information from its neighboring nodes. The underlying communication network is defined by $\mathcal{G} = (\mathbf{N},\mathcal{E},\mathcal{A})$, and information sharing occurs between two time steps as illustrated in Fig.~\ref{fig:time_step}, which is widely used in the existing optimization-based DSE framework \citep{yang2025state,wang2024distributed}.

\begin{defn}\label{def:DUIO}
    Let ${x}_i$ be the estimate of $x$ produced by observer node $\mathcal{O}_i$. The communication topology among the observer nodes $\{\mathcal{O}_i\}_{i\in\mathbf{N}}$ is modeled by $\mathcal{G} = (\mathbf{N},\mathcal{E},\mathcal{A})$, and each sampling interval admits $\nu$ communication iterations.  The collection $\{\mathcal{O}_i\}_{i\in\mathbf{N}}$ is defined to be a discrete-time DUIO for system \eqref{eq:sys} if, as $t\to \infty$, the estimation error at every node satisfies $\|{x}_i(t) - x(t)\|$ being bounded and these bounds converges to zero as the communication frequency $\nu \to \infty$.

\end{defn}
\section{Discrete-time DUIO Algorithm Design}\label{sec:DUIO}
\subsection{Local state estimation}\label{subsec:red-od:UIO}
\begin{figure}[htp] 
\centering
\scalebox{0.9}{
    \begin{tikzpicture}[
    >=latex, 
    font=\small, 
    node distance=1.5cm
]

    \def\xgap{2.8}   
    \def\ygap{1.8}   
    
    
    \node (W1) at (0, 2*\ygap) {$\mathscr{W}^*_{g,i}$};
    \node (W2) at (\xgap, 2*\ygap) {$\mathscr{W}^*_{g,i}$};
    
    \node (X1) at (0, \ygap) {$\mathscr{X}$};
    \node (X2) at (\xgap, \ygap) {$\mathscr{X}$};
    
    \node (Q1) at (0, 0) {$\mathscr{X} / \mathscr{W}^*_{g,i}$};
    \node (Q2) at (\xgap, 0) {$\mathscr{X} / \mathscr{W}^*_{g,i}$};
    
    \node (U) at (-2.0, \ygap) {$\bar{\mathscr{U}}_i$};

    
    \draw[->] (W1) -- (W2) node[midway, above] {$A_{L_i}|\mathscr{W}^*_{g,i}$};
    
    \draw[->] (X1) -- (X2) node[midway, above] {$A_{L_i}$};
    
    \draw[->] (Q1) -- (Q2) node[midway, above] {$A_{L_i}|\mathscr{X}/\mathscr{W}^*_{g,i}$};
    \node[below=2pt] at ($(Q1)!0.5!(Q2)$) {spectrum good}; 

    
    \draw[->] (W1) -- (X1) node[midway, right] {$W^*_{g,i}$};
    \draw[->] (X1) -- (Q1) node[midway, right] {$P_{W^*_{g,i}}$};
    
    \draw[->] (W2) -- (X2) node[midway, right] {$W^*_{g,i}$};
    \draw[->] (X2) -- (Q2) node[midway, right] {$P_{W^*_{g,i}}$};

    \draw[->] (U) -- (X1) node[midway, above] {$\bar{B}_i$};
    \draw[->, dashed] (U) -- (Q1) node[midway, below left=-2pt] {$0$};

    
    \def\braceX{\xgap + 1.2} 
    \def\braceAmp{8pt}
    
    \draw[decorate, decoration={brace, amplitude=\braceAmp}, thick]
        (\braceX, 2*\ygap + 0.2) -- (\braceX, \ygap - 0.2)
        node[midway, right=10pt, align=left] {Obtained by\\distributed\\optimization};

    \draw[decorate, decoration={brace, amplitude=\braceAmp}, thick]
        (\braceX, \ygap - 0.6) -- (\braceX, -0.2)
        node[midway, right=10pt, align=left] {Obtained by\\local state\\estimation};

\end{tikzpicture}}\\[-1.5ex]
    \caption{Commutative diagram of $\W_{g,i}^*$ decomposition at Node $i$.}
    \label{fig:Wg_decompose}
\end{figure}
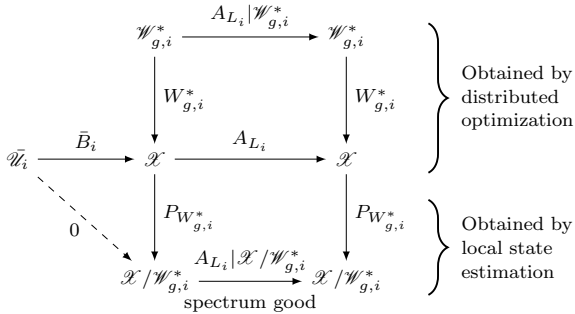
In this subsection, we design for each node $i$ a local state estimation that filters out the effect of the local unknown inputs while preserving the information that is still reconstructible from the local measurements. 
The construction follows the geometric approach in \cite{zhao2025DUIO,zhao2025bridging} and is briefly recalled here for completeness.\\
Fig.~\ref{fig:Wg_decompose} illustrates the basic idea. The state space $\X$ can be decomposed into a subspace, $\W_{g,i}^*$, that is directly corrupted by the unknown inputs and its complementary quotient space, $\X/\W_{g,i}^*$, on which the dynamics can be shaped to have a ``good'' (e.g, stable) spectrum for estimation. Intuitively, the goal is to project the state onto the largest subspace that is unaffected by the unknown inputs $\bar{u}_i$ at each individual node, and to design an observer that evolves on this quotient space.\\
We first recall an important lemma from \cite{zhao2025DUIO}.
\begin{lem}\label{lem:Wg}
    Let $P_{W_i^*}:\X \to \X/\W_i^*$ be the canonical projection, where $\W_i^*$ is the infimal $(C_i,A)$-invariant subspace containing $\image{\bar{B}_i}$.
    Then, the subspace $\W_{g,i}^*$ defined as
    \begin{equation}\label{eq:Wg}
        \W_{g,i}^*
        \coloneqq 
        P_{W_i^*}^{-1}\bar{\X}_{b,i}^* 
    \end{equation}
    is the infimal $(C_i, A)$-invariant subspace containing $\image \bar{B_i}$, while enabling the assignment of $\kappa(A_{L_i}|\X/\W_{g,i}^*)$\footnote{For simplicity, we use $A_{L_i}$ to represent $A+L_iC_i$.} into the partial complex plane $\mathbb{C}_g$, where $\bar{\X}_{b,i}^*$ is the sub-quotient space corresponding to unstable invariant zeros of system $(C_i,A,\bar B_i)$\footnote{Readers may refer to \cite[Equation~(7)]{zhao2025bridging} for the detailed definition of $\bar{\X}_{b,i}^*$.}.
\end{lem}
In words, $\W_{g,i}^*$ is the smallest $(C_i,A)$-invariant subspace that contains all directions affected by the unknown inputs, and such that the dynamics restricted to the quotient space $\X/ \W_{g,i}^*$ can be immune to the unknown inputs $\bar u_i$ while made well-conditioned (i.e., $\kappa(A_{L_i}|\X/\W_{g,i}^*) \in \mathbb{C}_g$). This makes it possible to separate the disturbance-corrupted parts from those that can be robustly estimated.


We then consider the following local estimation at each node:
 \begin{equation}\label{eq:si}
    s_i(t+1) = \bar A_{L_i}  s_i(t) - P_{W_{g,i}^*}L_i y_i(t) + P_{W_{g,i}^*}B_i u_i(t)
\end{equation}
with $\kappa(\bar A_{L_i}) \in \mathbb{C}_g$. 
Let $ T_i \coloneqq\matrices{P_{W_{g,i}^*} \\ C_i}$, and $\xi_i(t) \coloneq \matrices{s_i(t)\\y_i(t)}$
The next lemma states its convergence property.
\begin{lem}\label{lemma:s_i}
     With local estimator \eqref{eq:si}, the vector $\xi_i(t)$ asymptotically converges to $T_ix(t)$ as $t\rightarrow\infty$.
\end{lem}
The proof of Lemma~\ref{lemma:s_i} is omitted as it follows directly from the derivations in \cite{zhao2025DUIO,zhao2025bridging}.

\subsection{Data Fusion via Distributed Optimization}
Under the asymptotic convergence $\xi_i(t)\to T_ix(t)$ (as $t\to\infty$), the sensor nodes are capable of collaboratively reconstructing the global system state $x(t)$ using the distributed optimization-based algorithm proposed in this subsection. As a prerequisite, we introduce the following joint geometric condition.
\begin{assum}\label{asm:geo_con}
The canonical projections $P_{W_{g,i}^*}:\X\rightarrow\X/\W_{g,i}^*$ and the output matrices $C_i$ of all the node $i$ have the following joint property
\begin{equation}\label{con:con_for_dt_system}
        \bigcap_{i=1}^N \kernel \matrices{P_{W_{g,i}^*} \\ C_i} = 0.
\end{equation}
\end{assum}
\begin{rem}
    Assumption~\ref{asm:geo_con} guarantees global reconstructability from all local projections. In particular, it ensures that no nonzero state is simultaneously invisible to both the reduced-order estimator constructed in Subsection~\ref{subsec:red-od:UIO} and the direct local measurement across all nodes. 
    Hence, Assumption~\ref{asm:geo_con} replaces the restrictive per-node rank condition commonly imposed in existing DUIO designs. This rank condition is formulated by $\rank (C_i \bar B_i) = \rank \bar B_i, \forall i \in \mathbf{N}$, which requires each node to be individually capable of completely decoupling its unknown inputs. In contrast, our approach requires only the joint property in \eqref{con:con_for_dt_system} without enforcing such rigid structural conditions locally. This feature substantially broadens the applicability of the proposed DUIO design, particularly for networks in which individual nodes have limited sensing capabilities.
\end{rem}

Based on the local auxiliary estimation $\xi_i$, at time step $t$, we can formulate the distributed state estimation problem with local unknown inputs as follows:
\begin{equation}\label{eq:init_pro}
    \begin{aligned}
        &\min_{{x}_i^t\in\mathbb{R}^n, i\in\mathbf{N}} \sum_{i\in\mathbf{N}}f_i(x_i^t)\\
        &\qquad\mathrm{s.t.}\quad {x}_i^t={x}_j^t,\ (i,j)\in\mathcal{E},
    \end{aligned}
\end{equation}
where $f_i(x_i^t)$ is defined as \begin{equation}
    f_i(x_i^t)\coloneq\frac{1}{2}\| T_i {x}_i^t-\xi_i^t\|^2
\end{equation}
with $\xi_i^t \coloneq\xi_i(t)$. For simplicity, let
\begin{equation}
    F(X^t)\coloneq \sum_{i\in\mathbf{N}}f_i(x_i^t),
\end{equation}
where $X^t=\col(x_1^t,x_2^t,\cdots,x_N^t)$. Based on the definitions, for a decision variable $\mathtt{x}$, we can observe that $f_i(\mathtt{x})$ is $K_i$-Lipschitz continuous with $K_i=\|T_i^\top T_i\|$, and $\bm{f}(\mathtt{x})\coloneq\sum_{i=1}^N f_i(\mathtt{x})$ is $\mu$-strongly convex with $\mu = \lambda_{\mathrm{min}}\left(\sum_{i=1}^N T_i^\top T_i\right)$ and has inequality
\begin{equation}\label{eq:strong_convexity}
    \sum_{i=1}^N f_i(\mathtt{x})-\sum_{i=1}^N f_i(\mathtt{x}^\star)\ge \frac{\mu}{2}\|\mathtt{x}-\mathtt{x}^\star\|^2,
\end{equation}
where $\mathtt{x}^\star$ is the optimal solution of $\min_{\mathtt{x}}\ \bm{f}(\mathtt{x})$.
Under Assumption~\ref{asm:connected}, the constraint in \eqref{eq:init_pro} enforces
    \begin{equation*}
        x_1^t=x_2^t=\cdots=x_N^t.
    \end{equation*}
    Let $\mathtt{x}^t=x_1^t=x_2^t=\cdots=x_N^t$ denote this common variable. Then \eqref{eq:init_pro} reduces to the following unconstrained optimization problem
\begin{equation}\label{eq:pro_1}
        \begin{aligned}
        &\min_{\mathtt{x}^t\in\mathbb{R}^n}\ \bm{f}(\mathtt{x}^t)=\sum_{i\in\mathbf{N}}f_i(\mathtt{x}^t)=\sum_{i\in\mathbf{N}}\frac{1}{2}\| T_i  \mathtt{x}^t-\xi_i^t\|^2.
    \end{aligned}
\end{equation}
\begin{lem}
    Under Assumptions~\ref{asm:connected} and~\ref{asm:geo_con}, if there exist an algorithm solving problem \eqref{eq:init_pro} whose solution is equivalent to optimization problem \eqref{eq:pro_1}, then the outputs of this algorithm constitute a DUIO satisfying Definition~\ref{def:DUIO}.
\end{lem}
\begin{pf}
The first‐order optimality condition $\nabla \bm{f}(\mathtt{x}^{t\star}) = 0$ yields
\begin{equation*}
    \sum_{i=1}^N T_i ^\top\left( T_i \mathtt{x}^{t\star}-\xi_i^t\right)=0,
\end{equation*}
which can be rewritten as
\begin{equation}\label{eq:TTsum}
    \left(\sum_{i=1}^N T_i ^\top T_i \right)\mathtt{x}^{t\star} = \sum_{i=1}^N T_i ^\top\xi_i^t
\end{equation}
Assumption~\ref{asm:geo_con} states that $\bigcap_{i=1}^N\kernel{ T_i }=0$, which is equivalent to $\kernel{\col_{i\in\mathbf{N}} (T_i) }=0$. Hence, the matrix $\left(\sum_{i=1}^N T_i ^\top T_i \right)$ is positive definite and therefore invertible. Consequently, one can obtain
\begin{equation*}\label{eq:phi}
    \mathtt{x}^{t\star} = \left(\sum_{i=1}^N T_i ^\top T_i \right)^{-1}\sum_{i=1}^N T_i ^\top\xi_i^t
\end{equation*}
Moreover, Lemma~\ref{lemma:s_i} guarantees that
\begin{equation}\label{eq:tinf}
    \lim_{t\rightarrow\infty}\xi_i^t =  T_i  x(t),\ \forall i\in\mathbf{N}
\end{equation}
Substituting \eqref{eq:tinf} into \eqref{eq:TTsum} gives
\begin{equation*}
    \lim_{t\rightarrow\infty}\mathtt{x}^{t\star} =\left(\sum_{i\in\mathbf{N}} T_i ^\top T_i \right)^{-1}\sum_{i=1}^N T_i ^\top T_i  x(t)=  x(t)
\end{equation*}
Thus, each local estimate $x_i$ satisfies
\begin{equation}\label{eq:x_star}
    \lim_{t\rightarrow\infty} {x}_i(t) =\lim_{t\rightarrow\infty}\mathtt{x}^{t\star} =  x(t),\ \forall i \in\mathbf{N}
\end{equation}
which completes the proof. \hfill$\square$
\end{pf}
Subsequently, we employ a distributed optimization algorithm to solve \eqref{eq:init_pro}, where each individual node $i$ utilize its local input signals $u_i$, local measurements $y_i$ and estimates shared from its neighbors $\mathcal{N}_i$ to reconstruct the entire system state. Importantly, the solution obtained via this distributed approach is identical to that of the initial problem \eqref{eq:init_pro} and problem \eqref{eq:pro_1}. Specifically, we adapt the distributed linearized alternative direction method of multipliers \cite{aybat2017distributed}. The update steps\footnote{To distinguish between the system dynamics and the communication iterations, we use the superscript $\cdot^t$ to represent the time update $\cdot(t)$, and the square bracket $[\cdot]$ with the index $k$ to denote the inner communication iterations occurring between time steps.} for Node $i$ are defined as follows:
\begin{equation}\label{eq:opt:iter}
\left\{\begin{aligned}
    &x_i^t[k+1] = x_i^t[k] - \alpha_i\left(\nabla f\left(x_i^t[k]\right)+\phi_i^t[k]+\psi_i^t[k]\right),\\
    &\phi_i^t[k+1] = \frac{\gamma}{2}\sum_{i=1}^N a_{ij}\left(x_i^t[k+1]-x_j^t[k+1]\right),\\
    &\psi_i^t[k+1] = \psi_i^t[k] + \phi_i^t[k+1],
\end{aligned}\right.
\end{equation}
where $\gamma>0$ is the penalty parameter, and $\alpha_i$ is the step size given by
\begin{equation}\label{eq:step_size}
    \alpha_i = \frac{1}{K_i + \gamma d_i},
\end{equation}
where $K_i$ is the Lipschitz constant of $f_i(\cdot)$ and $d_i$ is degree of node $i$.

Then, we can build the complete discrete-time DUIO algorithm shown in Algorithm~\ref{alg:my_alg}.
At each time step $t$, we first propagate the local reduced order observer \eqref{eq:si} to obtain $s_i(t)$, and then run $\nu$ rounds of the distributed optimization iterations \eqref{eq:opt:iter} to solve \eqref{eq:init_pro}.
Furthermore, the convergence properties are investigated in Theorem~\ref{thm:main}.
\begin{algorithm2e}[htp]
\caption{The algorithm of distributed state estimation with local unknown inputs at node $i$.}
\label{alg:my_alg}
\KwIn{step size $\alpha_i$, penalty parameter $\gamma$, and communication frequency $\nu$.
}
Collect local measurement $y_i(0)$ and local input $u_i(0)$\\
Initialize local estimate $s_i(0)$.\\
\For{$t=1,2,3,\cdots$}{
    Calculate local reduced-order estimate
    $s_i(t) = \bar A_{L_i}  s_i(t-1) - P_{W_{g,i}^*}L_i y_i(t-1) + P_{W_{g,i}^*}B_i u_i(t-1)$.\\
    Collect $y_i(t)$ and $u_i(t)$\\
    Get $\xi^t_i=\begin{bmatrix}
        s_i(t)^\top & y_i(t)^\top
    \end{bmatrix}^\top$.\\
    Initialize ${x}_{i}^t[0]=\mathbf{0}_n$, $\phi_i^t[0]=\mathbf{0}_n$, $\psi_i^t[0]=\mathbf{0}_n$.\\
    $k \gets 1$\\
    \While{$k \le \nu$}{
    Calculate ${x}_i^t[k]$
    \begin{align*}
        {x}_i^t[k] = &\left(I_n\!-\!\alpha_i T_i^\top T_i\right){x}_i^t[k-1] \\&+ \alpha_i\left( T_i^\top\xi_i^t - \phi_i^t[k-1] - \psi_i^t[k-1]\right)
    \end{align*}\\
    Broadcast $x_i^t[k]$ to all neighbors $j\in\mathcal{N}_i$\\
    Receive $x_j^t[k]$ from all neighbors $j\in\mathcal{N}_i$\\
    Update $\phi_{i}^t[k]$ and $\psi_i^t[k]$
    \begin{equation*}
    \begin{aligned}
     &\phi_i^t[k] = \frac{\gamma}{2}\sum_{i=1}^N a_{ij}\left(x_i^t[k]-x_j^t[k]\right)\\
     &\psi_i^t[k] = \psi_i^t[k-1] + \phi_i^t[k]
    \end{aligned}
    \end{equation*}\\
    $k \gets k + 1$
    }
    Get the estimated system state
    \begin{equation*}
        {x}_i(t) = {x}_i^t[\nu]
    \end{equation*}
}
\KwOut{Estimated state sequence $\{{x}_i(t)\}_{t\ge 1}$}
\end{algorithm2e}
\begin{thm}\label{thm:main}
    Under Assumption~\ref{asm:connected} and~\ref{asm:geo_con}, as $t\rightarrow\infty$, the estimated states $x_i(t)=x_i^t[\nu]$, $i\in\mathbf{N}$, generated by the proposed discrete-time DUIO carried by Algorithm~\ref{alg:my_alg} satisfies  
       \begin{align}
        \text{1)}\   \frac{1}{\nu}\|\sum_{k=1}^\nu(x_i^t[k]&-x(t))\|\le \frac{1}{\sqrt{\nu\mu}}\left(\sum_{i=1}^N\frac{1}{\alpha_i}\right)^{\frac{1}{2}}\|x(t)\|\notag\\
        &+\frac{1}{\sqrt{\nu\lambda_2(\mathcal{L})}}\left(\frac{2}{\gamma}+\frac{1}{2}\left(\sum_{i=1}^N \frac{1}{\alpha_i}\|x(t)\|^2\right)\right),\label{eq:thm}
       \end{align}
    2) $x_i(t)$ exactly converges to $x(t)$ if $\nu\rightarrow \infty$.
\end{thm}
\begin{pf}
By \cite[Theorem~7]{aybat2017distributed}, the sequence $x_i^t[k]$ converges to the optimal solution of \eqref{eq:init_pro}, i.e., $x(t)$, as $k\rightarrow \infty$. We now derive the claimed bound on the iteration-averaged estimation error.

Define $\bar{x}_i^t[\nu]\coloneq\frac{1}{\nu}\sum_{k=1}^\nu x_i^t[k]$, $\bar{x}^t[\nu]\coloneq\frac{1}{N}\sum_{i=1}^N x_i^t[\nu]$, and $\bar{X}^t[\nu]\coloneq\mathbf{1}_N\otimes \bar{x}^t[\nu]$.
Using the bounds claimed in \cite[Theorem~7]{aybat2017distributed}, we can obtain the following inequalities
\begin{align}
    &F(\bar{X}^t)-F^\star\le \frac{1}{2\nu}\left(\sum_{i=1}^N \frac{1}{\alpha_i}\|x(t)\|^2\right),\label{eq:bound_1}\\
    &\left(\sum_{(i,j)\in\mathcal{E}}\|\bar{x}_i^t[\nu]-\bar{x}_j^t[\nu]\|^2\right)^{\frac{1}{2}}\le \frac{2}{\nu\gamma}+\frac{1}{2\nu}\left(\sum_{i=1}^N \frac{1}{\alpha_i}\|x(t)\|^2\right).\label{eq:bound_2}
\end{align}
Since $F^\star=0$ and $\bm{f}(\mathtt{x})$ is $\mu$-strongly convex according to \eqref{eq:strong_convexity}, we have that
\begin{equation}
    F(\bar{X}^t[\nu])=\bm{f}(\bar{x}^t[\nu])=\sum_{i=1}^N f_i(\bar{x}^t[\nu])\ge\frac{\mu}{2}\|\bar{x}^t[\nu]-x(t)\|^2.\label{eq:Fmu}
\end{equation}
Integrating \eqref{eq:Fmu} with \eqref{eq:bound_1} yields
\begin{equation}
    \|\bar{x}^t[\nu]-x(t)\|\le\frac{1}{\sqrt{\nu\mu}}\left(\sum_{i=1}^N\frac{1}{\alpha_i}\right)^{\frac{1}{2}}\|x(t)\|.\label{eq:ineq_1}
\end{equation}
By applying the Poincar\'e inequality \cite{li2012geometric}, we get
\begin{equation}
\begin{aligned}
        \sum_{(i,j)\in\mathcal{E}}\|\bar{x}_i^t[\nu]-\bar{x}_j^t[\nu]\|^2 &\ge \lambda_2(\mathcal{L})\sum_{i=1}^N\|\bar{x}_i^t[\nu]-\bar{x}^t[\nu]\|^2\\
        &\ge \lambda_2(\mathcal{L})\|\bar{x}_i^t[\nu]-\bar{x}^t[\nu]\|^2.
\end{aligned}\label{eq:intro_lambda}
\end{equation}
Combining \eqref{eq:intro_lambda} and \eqref{eq:bound_2}, one can obtain
\begin{equation}
    \|\bar{x}_i^t[\nu]-\bar{x}^t[\nu]\|\le \frac{1}{\sqrt{\nu\lambda_2(\mathcal{L})}}\left(\frac{2}{\gamma}+\frac{1}{2}\left(\sum_{i=1}^N \frac{1}{\alpha_i}\|x(t)\|^2\right)\right).\label{eq:ineq_2}
\end{equation}
Finally, by the triangle inequality,
\begin{equation}
    \|\bar{x}_i^t[\nu]-x(t)\|\le \|\bar{x}_i^t[\nu]-\bar{x}^t[\nu]\|+\|\bar{x}^t[\nu]-x(t)\|,\label{eq:trian_ineq}
\end{equation}
and substituting \eqref{eq:ineq_1} and \eqref{eq:ineq_2} into \eqref{eq:trian_ineq} yields
\begin{equation*}
\begin{aligned}
        \frac{1}{\nu}\|\sum_{k=1}^\nu(x_i^t[k]-&x(t))\|=\|\bar{x}_i^t[\nu]-x(t)\|\\
        &\le \frac{1}{\sqrt{\nu\mu}}\left(\sum_{i=1}^N\frac{1}{\alpha_i}\right)^{\frac{1}{2}}\|x(t)\|\\
        &+\frac{1}{\sqrt{\nu\lambda_2(\mathcal{L})}}\left(\frac{2}{\gamma}+\frac{1}{2}\left(\sum_{i=1}^N \frac{1}{\alpha_i}\|x(t)\|^2\right)\right),
\end{aligned}
\end{equation*}
which completes the proof.
    \hfill$\square$
\end{pf}
\begin{rem}
    The right-hand side of \eqref{eq:thm} consists of two components. The first term, stemming from the global $\mu$-strong convexity, characterizes the rate at which the iteration-averaged estimates converge toward the true state. The second term, determined by the network connectivity $\lambda_2(\mathcal{L})$ and the penalty parameter $\gamma$, quantifies the decay of the consensus error across the network.\\
    Although an explicit bound for the instantaneous estimation error $x_i(t)-x(t)$ is not explicit, \eqref{eq:thm} nonetheless reveals how the convergence of the iteration-averaged error $\frac{1}{\nu}\sum_{k=1}^\nu (x_i^t[k]-x(t))$ depends on the strong convexity constant $\mu$, the algebraic connectivity $\lambda_2(\mathcal{L})$, the communication/iteration frequency $\nu$, the step size choice $\alpha_i$, and the real time system state $x(t)$. \\
    Hence, for sparse networks with small algebraic connectivity or systems with poorly conditioned $\sum_{i=1}^N T_i^\top T_i$, a larger communication frequency $\nu$ and the preconditioning step become crucial to maintain a reasonable convergence rate.
\end{rem}
It is worth noting that the convergence rate of Algorithm~\ref{alg:my_alg} critically depends on the global strong convexity parameter $\mu = \lambda_{\mathrm{min}}\left(\sum_{i=1}^N T_i^\top T_i\right)$, which arises from the curvature of the aggregate objective $\bm{f}(\mathtt{x})\coloneq\sum_{i=1}^N f_i(\mathtt{x})$. When $\lambda_{\min} \left( \sum_{i} T_i^\top T_i \right) \ll 1$, the strong convexity modulus $\mu$ becomes very small.
Since the convergence upper bound \eqref{eq:thm} guarantees scale inverse-proportionally with $\sqrt{\mu}$, a small global curvature directly results in significantly degraded convergence speed.\\
To mitigate this limitation, it is necessary to introduce an appropriate normalization step using Cholesky factorization \cite[Corollary~7.2.9]{horn2012matrix} to enhance the effective strong convexity seen by the algorithm.\\
Let $H = S^\top S$ be the Cholesky factorization of $H=\sum_{i=1}^N T_i^\top T_i$. Then, we can define the transformed local matrices
    \begin{equation}\label{eq:Ttilde}
    \widetilde T_i := T_i S^{-1}.
    \end{equation}
Accordingly, the transformed optimization problem can be rewritten as follows
    \begin{equation}\label{eq:trans_pro}
    \min_{\mathtt{z}\in\mathbb{R}^n}\ \widetilde{\bm f} (\mathtt{z})
    := \sum_{i=1}^N \widetilde f_i(\mathtt{z}),
    \end{equation}
where $\widetilde f_i(\mathtt{z})
    := \frac12 \|\widetilde T_i \mathtt{z} - \xi_i\|^2$.
Then we have the following lemma holds.
\begin{lem}\label{lemma:precondition}
    The optimal solution $\mathtt{z}^\star$ of the optimization problem \eqref{eq:trans_pro} has the following relationship with $\mathtt{x}^\star$
    \begin{equation}
        \mathtt{x}^\star=S^{-1}\mathtt{z}^\star,
    \end{equation}
    and the function $\widetilde{\bm{f}}(\mathtt{z})$ is 1-strongly convex.
\end{lem}
\begin{pf}
    From \eqref{eq:Ttilde}, we have $\widetilde T_i = T_i S^{-1}$, which yields
\begin{equation*}
    \sum_{i=1}^N \widetilde T_i^\top \widetilde T_i
    = S^{-\top}\Big(\sum_{i=1}^N T_i^\top T_i\Big)S^{-1}
    = S^{-\top} H S^{-1}
    = I_n.
\end{equation*}
Hence $\nabla^2\widetilde{\bm f}(\mathtt{z}) = I_n$, which implies $1$--strong convexity. 

Moreover, since $z_i = Sx_i$ is a bijection, we have
\begin{equation*}
    \widetilde{\bm f}(\mathtt{z})
    = \sum_{i=1}^N \frac12\|T_i S^{-1}\mathtt{z} - \xi_i\|^2
    = {\bm f}(S^{-1}\mathtt{z}).
\end{equation*}
Therefore
\begin{equation*}
    \mathtt{z}^\star = \arg\min_{\mathtt{z}} \widetilde{\bm f}(\mathtt{z})
    \iff
    S^{-1}\mathtt{z}^\star = \arg\min_{\mathtt{x}} \bm{f}(\mathtt{x})=\mathtt{x}^\star,
\end{equation*}
which completes the proof.
\hfill$\square$
\end{pf}
With Lemma~\ref{lemma:precondition}, we can replace $T_i$ with $\widetilde T_i$ in Algorithm~\ref{alg:my_alg} to calculate $z_i(t)$ and then obtain $x_i(t)$ by $x_i(t)=S^{-1}z_i(t)$ as illustrated as follows
\begin{equation*}
\left\{
\begin{aligned}
    &\text{Replace}\ T_i\ \text{with}\ \widetilde T_i\\
    &\text{Normalized Algorithm~\ref{alg:my_alg} outputs}\ z_i(t)\\
    &\text{Obtain}\ x_i(t)\ \text{by}\ x_i(t)=S^{-1}z_i(t)
\end{aligned}\right.\tag{PS}\label{preconditioned_step}
\end{equation*}

By transforming the original variables through a suitable linear mapping (e.g., based on the Cholesky factor of $\sum_{i=1}^N T_i^\top T_i$), the global curvature can be normalized, the condition number of the objective can be substantially improved, and consequently the algorithm iterations achieve significantly faster practical and theoretical convergence. This advantage of the preconditioning step is fully reflected in the simulation results in Section~\ref{sec:simulate}.

\section{Simulation Results}\label{sec:simulate}
Consider a discrete-time LTI system with system matrix and input matrix as follows
\begin{equation*}
\setlength{\arraycolsep}{2.5pt}
\begin{aligned}
    &\ A =\\
    &\begin{bmatrix}
    0.9925&    0.1496&    0.0037&    0.0001&         0&         0&\\
   -0.0998&    0.9925&    0.0498&    0.0024&         0&         0\\
         0&         0&    0.9928&    0.0949&         0&         0\\
         0&         0&   -0.1424&    0.8978&         0&         0\\
    0.0002&   -0.0056&   -0.0037&    0.0472&    0.9850&   -0.1493\\
   -0.0037&    0.0744&    0.0016&    0.0049&    0.1990&    0.9850
    \end{bmatrix},
\end{aligned}
\end{equation*}
\begin{equation*}
\begin{aligned}
&\ B^\top =
\begin{bmatrix}
    \bm b_a&\bm b_b&\bm b_c
\end{bmatrix}^\top
=\\
&\begin{bmatrix}
    0.0037&    0.0499&         0&         0&    0.0497&    0.0069\\ \hdashline
         0&         0&    0.0024&    0.0475&    0.0012&    0.0001\\\hdashline
    0.0001&    0.0012&    0.0499&   -0.0036&   -0.0038&    0.0498
\end{bmatrix}.
\end{aligned}
\end{equation*}
The system input signal is given by \[ u(t) = \begin{bmatrix} \mathtt{sin}(0.1t) & \mathtt{cos}(0.5t) & 0.5\mathtt{sin}(0.5t) \end{bmatrix}^\top, \] which can be separated as \( u = \begin{bmatrix} u_a & u_b & u_c \end{bmatrix}^\top \). The partial input signals known by each individual node are given by \( u_1 = u_b \), \( u_2 = \begin{bmatrix} u_a & u_c \end{bmatrix}^\top \), \( u_3 = u_c \), and \( u_4 = \begin{bmatrix} u_b & u_c \end{bmatrix}^\top \). Accordingly, the known and unknown input channels for each node are specified as follows
\begin{equation*}
    \begin{aligned}
        &B_1 = \bm b_b = \bar{B}_2,\ \bar{B}_1=\matrices{\bm b_a&\bm b_c} = B_2,\\
        &B_3 = \bm b_c,\ \bar{B}_2 = \matrices{\bm b_a&\bm b_b},\ 
        B_4 = \matrices{\bm b_b&\bm b_c},\ \bar B_4=\bm b_a.
    \end{aligned}
\end{equation*}
The whole system outputs are measured by $4$ sensors with the output matrices
\begin{equation*}
    \begin{aligned}
        &C_1=\begin{bmatrix}
        1& 0& 0& 0& 0& 0
        \end{bmatrix},\ 
        C_2=\begin{bmatrix}
        0& 1& 0& 0& 0& 0
        \end{bmatrix},\\
        &C_3=\begin{bmatrix}
        0& 0& 1& 0& 0& 0
        \end{bmatrix},\ 
        C_4=\begin{bmatrix}
        0& 0& 0& 1& 0& 1
        \end{bmatrix}.
    \end{aligned}
\end{equation*}
The underlying communication network of the sensor nodes are given by the graph illustrated in Fig.~\ref{fig:topo_for_num_exam}.
\begin{figure}
    \centering
    \scalebox{0.8}{\begin{tikzpicture}
\def\off{18}
\def\N{5}
\def\R{2.5}
\pgfmathparse{360/\N}
\edef\step{\pgfmathresult}

\colorlet{net}{teal}
\tikzset{comm/.style = {color=net, very thick, dash pattern=on 4pt off 1.5pt}}
\node [circle, 
        draw, 
        color=net, 
        fill=white, 
        text=black, 
        very thick,
        inner sep=6pt] (O1) at (0,2) {$\mathcal O_{1}$};
\node [circle, 
        draw, 
        color=net, 
        fill=white, 
        text=black, 
        very thick,
        inner sep=6pt] (O2) at (2,2) {$\mathcal O_{2}$};
\node [circle, 
        draw, 
        color=net, 
        fill=white, 
        text=black, 
        very thick,
        inner sep=6pt] (O3) at (2,0) {$\mathcal O_{3}$};
\node [circle, 
        draw, 
        color=net, 
        fill=white, 
        text=black, 
        very thick,
        inner sep=6pt] (O4) at (0,0) {$\mathcal O_{4}$};

\draw[comm] (O1) to[out=40, in=140] (O2);
\draw[comm] (O2) to[out=-40, in=40] (O3);
\draw[comm] (O3) to[out=-140, in=-40] (O4);
\draw[comm] (O4) to[out=140, in=-140] (O1);

\end{tikzpicture}}
    \caption{Communication network for numerical example.}
    \label{fig:topo_for_num_exam}
\end{figure}
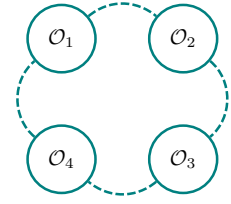

Based on the geometric approach \cite{zhao2025DUIO,zhao2025bridging}, we can compute $L_i$ and $P_{W_{g,i}^*}$ as follows
\begin{align*}
        &L_1^\top=\matrices{0&     0&     0&     0&     0&     0},\ L_3^\top=\matrices{0&    0&   -1.9754&    0.0998&   0&    0},\\
        &L_2^\top\!=\!\matrices{-0.2244&   -0.9990&    0.0252&   -0.0090&    0.4797&   -2.0659}\!,\\
        &L_4^\top=\matrices{-1.0799&    0.0809&    0&    0&    0.2710&   -1.9582},
\end{align*}
\begin{align*}
        &P_{W_{g,1}^*}=\matrices{0&     0&     0&     0&     0&     0},\\ &P_{W_{g,2}^*}=\matrices{0.0984&    0.9939&   -0.0491&    0.0025&   0&   0},\\
        &P_{W_{g,3}^*}=\matrices{0&    0&    0.9987&   -0.0505&   0&    0},\\
        &P_{W_{g,4}^*}=\matrices{ 0&    0&    0&    1&    0&    0\\
        0&    0.0035&   -0.9957&    0&    0.0093&   -0.0923\\
       0&   -0.0375&   -0.0928&    0&   -0.0998&    0.9900}.
\end{align*}
\begin{figure}[htp]
    \centering
    \includegraphics[width=\linewidth]{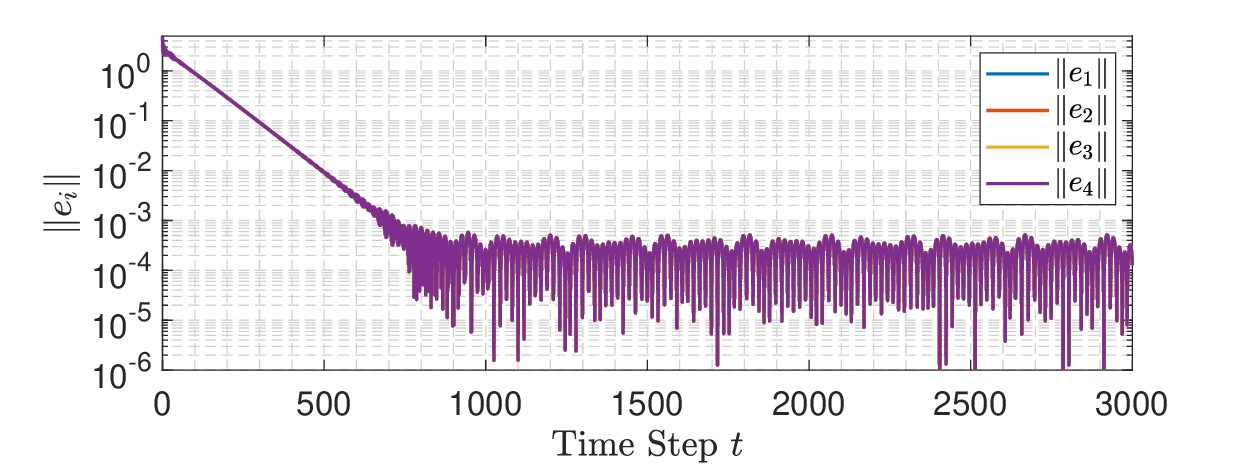}\\[-2.5ex]
    \caption{Estimation errors using initial Algorithm~\ref{alg:my_alg} with 30000 iterations per time step.}
    \label{fig:estimation_error_1}
\end{figure}
\begin{figure}[htp]
    \centering
    \includegraphics[width=\linewidth]{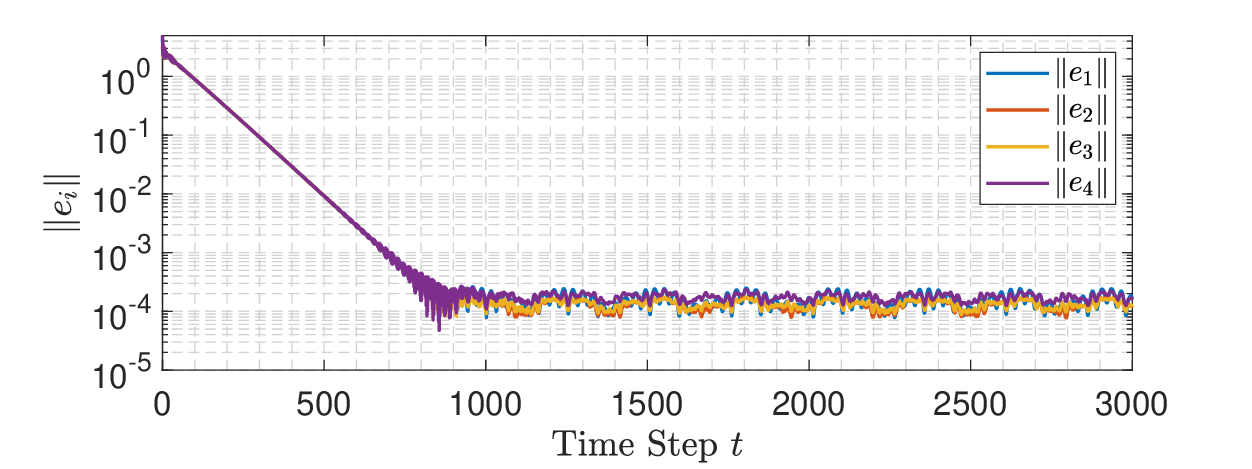}\\[-2.5ex]
    \caption{Estimation errors using normalized Algorithm~\ref{alg:my_alg} with \eqref{preconditioned_step} and 50 iterations per time step.}
    \label{fig:estimation_error_2}
\end{figure}
\begin{figure}[htp]
    \centering
    \includegraphics[width=\linewidth]{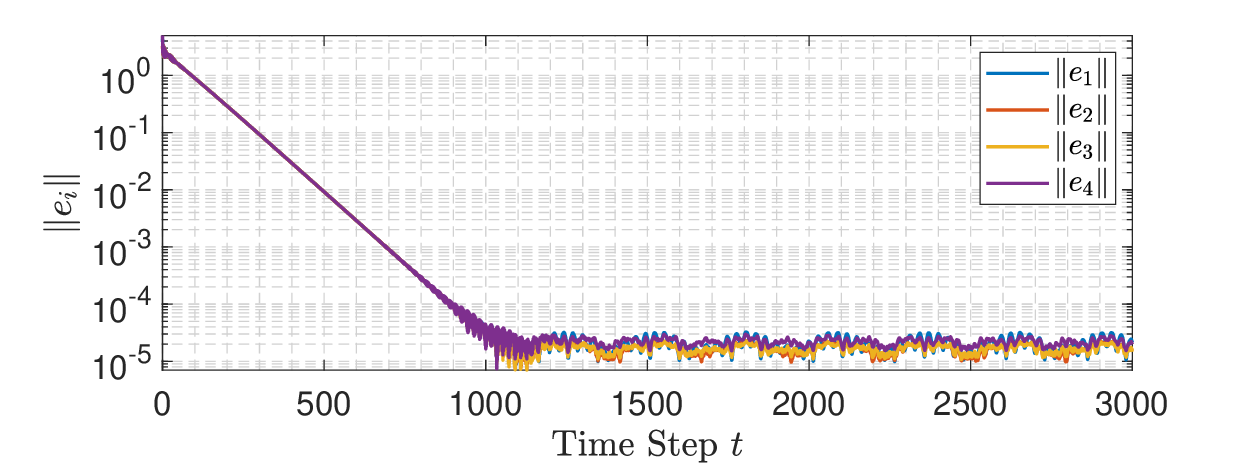}\\[-2.5ex]
    \caption{Estimation errors using normalized Algorithm~\ref{alg:my_alg} with \eqref{preconditioned_step} and 60 iterations per time step.}
    \label{fig:estimation_error_3}
\end{figure}
With the penalty parameter fixed at $\gamma=0.5$ and the step size $\alpha_i$ derived from \eqref{eq:step_size}, we can verify the performance of the proposed framework.

Fig.~\ref{fig:estimation_error_1} illustrates the convergence of estimation errors for the initial Algorithm~\ref{alg:my_alg}. While the error remains bounded, the convergence is inexact due to the finite $\nu$, corroborating our main theorem. The convergence of the consensus part of initial Algorithm~\ref{alg:my_alg}, requiring 30,000 iterations per time step, is attributed to the system's poor conditioning. Specifically, the eigenvalues of $\sum_{i=1}^N T_i^\top T_i$ are $\{0.0034,\ \allowbreak0.9970,\ \allowbreak1.0022,\ \allowbreak1.9966,\ \allowbreak2.9658,\ \allowbreak3.0351\}$. The presence of a eigenvalue that is much smaller than the other eigenvalues of $\sum_{i=1}^N T_i^\top T_i$, which leads to a small convexity constant $\mu$ and severely slow convergence rate of consensus part requiring $\nu=30000$.

However, by employing Algorithm~\ref{alg:my_alg} with normalized steps stated in \eqref{preconditioned_step}, we achieve similar estimation accuracy ($10^{-4}$ level) with only 50 iterations per time step, i.e. $\nu=50$, as shown in Fig.~\ref{fig:estimation_error_2}. This comparison demonstrates the practical necessity and effectiveness of the normalized step \eqref{preconditioned_step} for the distributed state estimation problem. Furthermore, by increasing the communication frequency to $\nu=60$, the estimation errors diminish to the $10^{-5}$ level as illustrated in Fig.\ref{fig:estimation_error_3}.

\section{Conclusion}\label{sec:conclude}
This paper proposes a novel two-time-scale framework for discrete-time Distributed Unknown Input Observer (DUIO) design, formulated as a distributed optimization problem. Crucially, our design eliminates the restrictive rank condition typically required for disturbance decoupling in the existing literature. In this framework, each local node estimates the maximum partial state based on local input and output information; subsequently, all sensor nodes collaboratively reach consensus via a distributed optimization algorithm to reconstruct the global system state. Furthermore, a preconditioning step is introduced to accelerate the convergence rate in scenarios where the strong convexity constant is small due to ``ill-conditioned'' system parameters. Simulation results validate the effectiveness of the proposed design and demonstrate the performance benefits of the preconditioning step. Future work will focus on extending this distributed optimization-based DUIO method to nonlinear systems and considering communication uncertainties \cite{bastianello2020asynchronous,bastianello2024robust}.


\end{document}